\documentclass[journal]{IEEEtran}

\usepackage[latin1]{inputenc}
\usepackage{graphicx, amssymb, amsmath, amsfonts, amsthm}
\usepackage[usenames,dvipsnames]{color}

\newtheorem{remark}{Remark}
\newtheorem{proposition}{Proposition}

\newcommand{\Rant}{R_{\text{ant}}}
\newcommand{\Vdc}{V_{\text{DC}}}
\newcommand{\fcut}{f_{\rm cut}}
\newcommand{\Df}{\Delta\!f}

\allowdisplaybreaks
\IEEEoverridecommandlockouts

\begin{document}
\title{RC Filter Design for Wireless Power Transfer:\\ A Fourier Series Approach}

\author{Constantinos Psomas, \IEEEmembership{Senior Member, IEEE}, and Ioannis Krikidis, \IEEEmembership{Fellow, IEEE}\vspace*{-5mm}
\thanks{C. Psomas and I. Krikidis are with the Department of Electrical and Computer Engineering, University of Cyprus, Cyprus (email: \{psomas, krikidis\}@ucy.ac.cy). This work has received funding from the European Research Council (ERC) under the European Union's Horizon 2020 research and innovation programme (Grant agreement No. 819819).}}

\maketitle

\begin{abstract}
In this letter, we study the impact of the low-pass resistor-capacitor (RC) filter on radio frequency (RF) wireless power transfer (WPT). The RC filter influences both the RF bandwidth by removing the harmonics as well as the ripple voltage at the output of the rectifier. In particular, a large (small) RC time constant, reduces (increases) the ripple but decreases (enhances) the direct-current (DC) component. By following a Fourier series approach, we obtain closed-form expressions for the rectifier's output voltage, the RC filter's output as well as the DC voltage. Our analytical framework provides a complete characterization of the RC filter's impact on the WPT performance. We show that this complete and tractable analytical framework is suitable for the proper design of the RC filter in WPT systems.
\end{abstract}

\begin{IEEEkeywords}
RC filter, rectenna modeling, wireless power transfer.\vspace{-2mm}
\end{IEEEkeywords}

\section{Introduction}
Far-field wireless power transfer (WPT) relates to energy harvesting via radio-frequency (RF) signals, where in contrast to conventional energy harvesting techniques (i.e., from renewable sources), it is a continuous, controllable and on-demand process \cite{NS}. Its efficiency relies on a appropriate end-to-end design of both the transmitter and the receiver. Strictly speaking, the aim is to increase the ratio between the RF power harvested at the receiver over the one emitted by the transmitter. An important element of the receiver's rectifier is the low-pass resistor-capacitor (RC) filter \cite{PAN, SC}. Indeed, the RC filter limits the RF bandwidth of the circuit by removing the harmonics but also controls the ripple voltage at the rectifier's output.

However, existing works do not discuss the impact of the RC filter on the WPT performance but mainly focus on the waveform design and the non-linearities of the rectifying circuit. Particularly, the work in \cite{AG} shows experimentally that waveforms with high peak-to-average power ratio (PAPR) increase the rectifier's RF to direct-current (DC) conversion efficiency. A theoretical study on multisine waveform design for WPT is considered in \cite{BC}, which proposes a non-linear rectenna model and, by assuming perfect channel state information, uses it to design multisine waveforms that increase the WPT efficiency. In \cite{IK}, a low-complexity tone-index modulation technique is proposed, which exploits multisine waveforms for WPT and embeds information in the number of tones. Also, the work in \cite{RM}, focuses on the waveform's optimal input distribution that maximizes the information transfer conditioned on the minimum RF harvesting. However, the aforementioned theoretical studies assume that the RC time constant is very large and neglect its effect on the energy harvesting. On the other hand, some efforts were made in \cite{PAN} to study the effects of the RC filter. Nevertheless, these were shown experimentally and no theoretical investigation has been undertaken so far.

Motivated by this, in this paper, we develop a mathematical framework to address this gap. In particular, we deal with a simple point-to-point WPT setup, where the destination harvests energy from the source's RF signal. Our focus lies on rectenna modeling and the impact of the low-pass RC filter on the energy harvesting. Following a mathematical approach based on Fourier analysis, allows us to derive closed-form expressions for the rectifier's output voltage, the RC filter's output as well as the DC voltage. Our analysis accurately captures the trade-off between the ripple voltage and the DC component in terms of the RC time constant. To the best of our knowledge, the developed analytical framework is the first in the WPT literature and it is appropriate for a thorough RC filter analysis and design for WPT systems.

\begin{figure}\centering
  \includegraphics[width=0.8\linewidth]{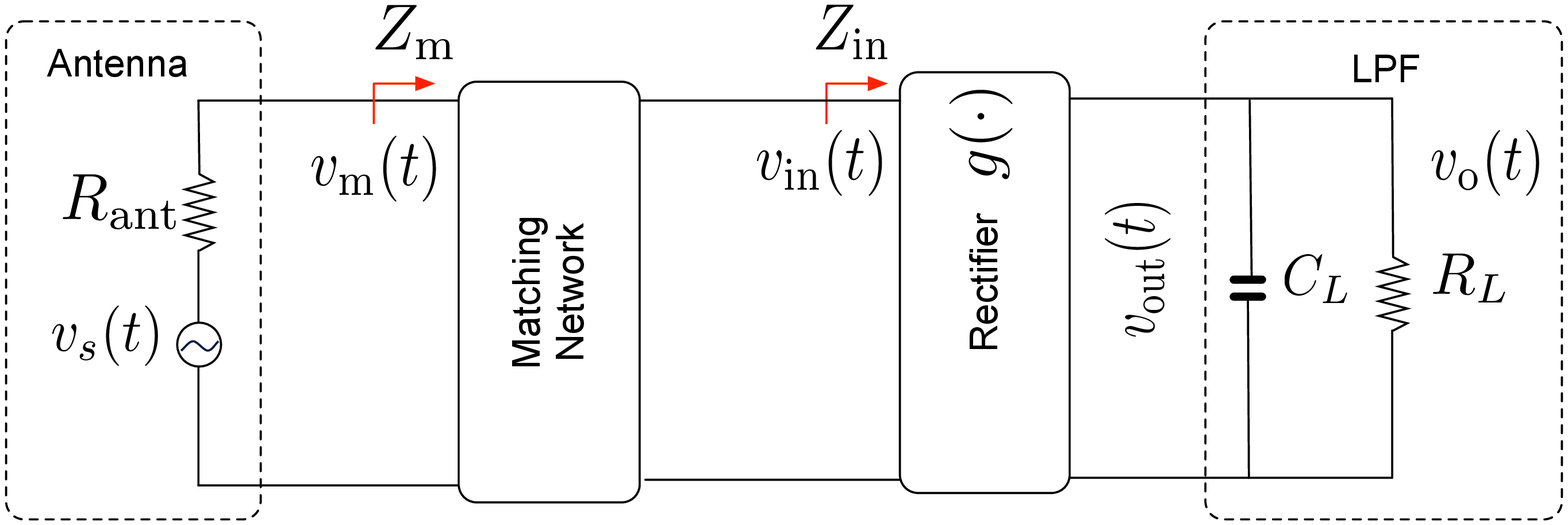}\vspace*{-1mm}
  \caption{The rectenna model consisting of an antenna, a matching network, a rectifier and a low-pass RC filter.}\label{fig1}\vspace*{-1mm}
\end{figure}

\section{Rectenna Model}

Consider a basic point-to-point WPT system, where the source aims to wirelessly transfer energy to a destination. The destination harvests energy from the source's RF signal through the employment of a rectenna; we consider the rectenna model illustrated in Fig. \ref{fig1}. The source transmits a sinewave signal
\begin{align}
  x(t) = A\cos(2\pi f_c t), ~ 0 \leq t \leq T,
\end{align}
where $f_c$ denotes the carrier frequency, $A$ is the signal's amplitude and $T$ is the duration of the waveform. The transmitted signal passes through a flat fading channel with channel gain $h$ and phase shift $\theta$. Thus, the received signal at the destination can be expressed as 
\begin{align}
  y(t) = |h| A \cos(\omega_c t + \theta),
\end{align}
where $\omega_c = 2 \pi f_c$ is the angular frequency; in this work, we will consider $h = 1$ and $\theta = 0$, without loss of generality\footnote{By keeping $h$ and $\theta$ fixed, allows for the derivation of simple closed-form expressions. The general case is straightforward but requires the use of their distribution functions.}.

The received signal is forwarded to a complex-conjugate matching network, which attempts to match the rectifier input impedance $Z_m$ to the antenna impedance $\Rant$, in order to increase the power transfer; in what follows, we assume perfect matching, i.e. $\Rant = Z_m$. In other words, we consider a theoretical bound, where the maximum possible power transfer is achieved \cite{SC, RM}. Then, by using the average power conservation law, the average received power is equal to the rectifier's input, written as
\begin{align}
\frac{1}{T}\int_0^T |y(t)|^2 dt &= \Re\left\{\frac{1}{T}\int_0^T \frac{|v_m(t)|^2}{\Rant} dt\right\}\nonumber\\
&= \Re\left\{\frac{1}{T}\int_0^T \frac{|v_{\text{in}}(t)|^2}{Z_{\text{in}}^*} dt\right\},
\end{align}
where $v_m(t)$ and $v_{\text{in}}(t)$ are the voltage signals before and after the matching network, respectively, $Z_{\text{in}}$ is the input impedance after the matching network, $\Re\{\cdot\}$ denotes the real part operator, and $Z^*$ is the conjugate of the complex number $Z$. As such, it follows that
\begin{align}
  y(t) = \frac{v_m(t)}{\sqrt{\Rant}}=v_{\text{in}}(t)\sqrt{\Re\{1/Z_{\text{in}}^*\}},
\end{align}
where we have
\begin{align}
  v_m(t)=\frac{v_s(t)}{2},
\end{align}
due to perfect matching \cite{RM}. Then, the input at the rectifier is
\begin{align}
  v_{\text{in}}(t) = y(t)\sqrt{\Re\{Z_{\text{in}}^*\}}.
\end{align}
For an idea diode and a parallel RC low-pass filter (LPF), we have \cite{SC,SC2}
\begin{align}
  Z_{\text{in}} = \frac{R_L}{1+\jmath\omega_c \tau},
\end{align}
where $\jmath = \sqrt{-1}$ is the imaginary unit, and so
\begin{align}
  \Re\{Z_{\text{in}}^*\} = \frac{R_L}{1+\omega_c^2 \tau^2},
\end{align}
where $\tau = R_L C_L$ is the RC time constant and $R_L$ and $C_L$ are the filter's load resistance and capacitor, respectively. Hence, the input voltage at the rectifier becomes equal to 
\begin{align}
  v_{\text{in}}(t) = \delta A \cos(\omega_c t),
\end{align}
with
\begin{align}
  \delta \triangleq \sqrt{\frac{R_L}{1+\omega_c^2 \tau^2}}.
\end{align}

It is worth noting that the matching network acts as a passive voltage amplifier for the received signal by utilizing a resonator with a high quality factor \cite{RM,TLE}; as such, $\delta$ refers to the amplification factor.

Finally, the diode-based circuit is modeled by two non-linear functions $g(x)=|x|$ (full-wave rectifier) and $g(x)=\max(0,x)$ (half-wave rectifier) \cite{SC, SC2}. The full-wave rectifier enables rectification during the entire waveform duration, whereas half is achieved with the half-wave rectifier. Despite their simplicity, these functions serve as a useful guideline for the RC filter design. Nevertheless, our analytical approach can be adapted to other more complex functions as well.


\section{RC Filter Design}
In this section, we provide the main results of our work. A validation of the RC filter's impact on the WPT performance is firstly presented through circuit simulations. Then, the analytical framework based on Fourier analysis is described.

\begin{figure}\centering
	\includegraphics[width=0.7\linewidth]{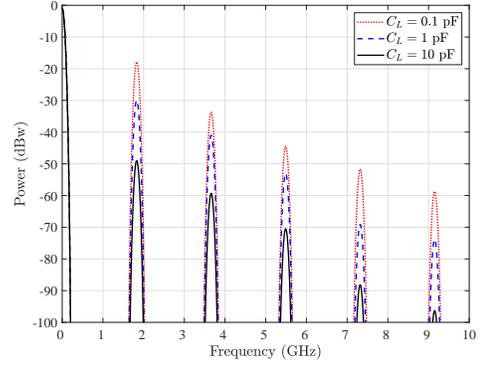}
	\caption{Frequency response of the harvested energy by exciting a single-tone.}\label{spectrum}\vspace*{-1mm}
\end{figure}

\subsection{Circuit Simulations}\label{circ_sim}

In order to verify the effects of the RC circuit, we provide simulation results of a diode bridge rectifier circuit by using LTspice \cite{LTS}. The simulated circuit implements the well-known full-wave bridge rectifier\footnote{The choice of rectifier is made mainly as an example for illustrating the RC filter's impact but other more efficient circuits could also be considered.}. The bridge is composed out of four SMS-7630 Schottky diodes; these diodes are ideal as they operate at low input powers \cite{BC,RM}. To evaluate the effect of the RF bandwidth, the serial resistance is set to $\Rant = 50 ~\Omega$ and the load resistance to $R_L = 2 ~{\rm k}\Omega$. Moreover, we consider $f_c = 915$ MHz and $A=1$ V. For the sake of consistency, we assume that the circuit is perfectly matched, and therefore, the matching circuit is omitted. Fig. \ref{spectrum} plots the frequency response of the harvested energy measured at the RC filter when a single-tone is fed to the circuit for different values for the capacitance. We can see that choosing a proper capacitance $C_L$ for a given load resistor, the several intermodulation terms may be effectively filter out, and hence, the ripples are also decreased. This provides the impetus for the analytical framework given in the next subsection.

\subsection{Fourier Series Analysis}
The output of the rectifier is a periodic, real function and therefore its Fourier series representation can be written as follows \cite{ISG}
\begin{align}
  v_{\text{out}}(t) &= g(v_{\text{in}}(t))\nonumber\\
  &=\delta A \left(\frac{a_0}{2}+\sum_{k=1}^{\infty}[a_k \cos(\omega_ckt)+b_k\sin(\omega_ckt)] \right)\nonumber\\
  &=\delta A \left(\frac{a_0}{2}+\sum_{k=1}^\infty d_k \cos(\omega_c k t+ \phi_k) \right),\label{eq1}
\end{align}
where $a_k$ and $b_k$ are the Fourier coefficients of $v_{\text{out}}(t)$, $d_k = \sqrt{a_k^2+b_k^2}$, and $\phi_k = \tan^{-1}(-b_k/a_k)$. The coefficients for the considered diode models are evaluated below.

\begin{proposition}\label{prop1}
The Fourier coefficients $a_k$ for the full-wave rectifier are
\begin{align}\label{coeff_full}
a_1 = 0, ~ a_k = \frac{4}{\pi(1-k^2)} \cos\left(\frac{\pi k}{2}\right), ~ k \neq 1,
\end{align}
and for the half-wave rectifier are
\begin{align}\label{coeff_half}
a_1 = \frac{1}{2}, ~ a_k = \frac{2}{\pi(1-k^2)} \cos\left(\frac{\pi k}{2}\right), ~ k \neq 1.
\end{align}
In both cases, $b_k = 0$, $\forall \,k$.
\end{proposition}

\begin{proof}
See Appendix.
\end{proof}

From the above proposition, it follows that $d_k = |a_k|$ and $\phi_k = 0$ if $a_k > 0$ and $\phi_k = \pi$, otherwise. It is important to mention that $a_k < 0$ when $k = 4n, n \in \mathbb{Z}^+$.

The low-pass RC filter limits the RF bandwidth of the circuit by removing the harmonics of $v_{\text{out}}(t)$ and controls the ripples of the signal at the time domain. The parallel RC filter can be modeled by the transfer function
\begin{align}
  H(f) = \frac{R_L}{1+\jmath 2\pi f R_L C_L}.
\end{align}
The filter's effect is illustrated in Fig. \ref{fig2}, where $\fcut = 1/(2\pi\tau)$ is the filter's cut-off frequency, that is, the harmonics with frequencies larger than $\fcut$ get (ideally) greatly attenuated. By taking into account the Fourier series representation of the filter's input in \eqref{eq1}, the output of the RC filter is written as 
\begin{align}
  v_{\rm o}(t) = \delta A \bigg(\frac{a_0 R_L}{2} +& \sum_{k=1}^\infty |H(k f_c)| d_k\nonumber\\  &\times \cos(\omega_c k t + \phi_k + \angle{H(k f_c)})\bigg),\label{vo}
\end{align}
where $|H(k f_c)| = R_L/\sqrt{1 + (2\pi kf_c\tau)^2}$ and $\angle{H(k f_c)} = \tan^{-1}(-2\pi kf_c\tau)$ is the magnitude and the argument of $H(k f_c)$, respectively. The peak of the time-domain ripple\footnote{The ripple refers to the variation in DC voltage at the output of the rectifier.} $\rho$ at the filter's output can be obtained by
\begin{align}\label{peak_ripple}
\rho = \delta A R_L\left(\frac{a_0}{2} + \sum_{k=1}^\infty \frac{d_k \cos(\phi_k)}{\sqrt{1 + (2\pi kf_c\tau)^2}}\right),
\end{align}
which occurs at time instants $t=\tan^{-1}(-2\pi kf_c\tau)/(\omega_c\tau)$, which gives $\cos(\phi_k) = 1$ if $a_k > 0$ and $\cos(\phi_k) = -1$, otherwise.

Finally, the DC voltage is the average $v_o(t)$ over the waveform duration, that is,
\begin{align}
  \Vdc = \mathbb{E}_t\{v_o(t)\} = \delta A R_L\frac{a_0}{2},
\end{align}
where $a_0$ is the Fourier coefficient that contributes to the DC component, which is $a_0 = 4/\pi$ (full-wave) and $a_0 = 2/\pi$ (half-wave). Thus, as expected, the DC voltage achieved by the full-wave is double of that achieved by the half-wave rectifier.

\begin{figure}\centering
  \includegraphics[width=0.8\linewidth]{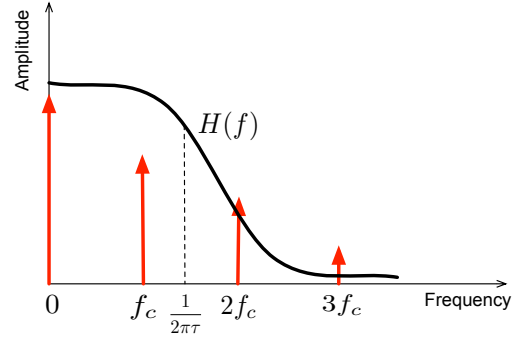}
  \caption{Low-pass RC filter with cut-off frequency $\fcut=1/(2\pi\tau)$.}\label{fig2}
\end{figure}

\begin{remark}
When the RC time constant is very large, i.e. $\tau \to \infty$, we have $\delta \to 0$. It follows that
\begin{align}\label{tauinf}
\lim_{\tau \to \infty} \Vdc = 0.
\end{align}
On the other hand, a small RC time constant, i.e. $\tau \to 0$, results in $\delta \to \sqrt{R_L}$. Therefore,
\begin{align}
\lim_{\tau \to 0} \Vdc = A R_L \sqrt{R_L} \frac{a_0}{2}.
\end{align}
\end{remark}

This remark highlights the importance of the RC filter's impact on energy harvesting, even though most works in the literature assume that $\tau\to\infty$. Note that, in practice, the output DC voltage when $\tau \to \infty$ is non-zero but negligible. Also, observe that for the case $\tau \to 0$, we obtain the maximum DC voltage but also the maximum time-domain ripple, given by
\begin{align}\label{max_ripple}
\rho_{\rm max} = AR_L \sqrt{R_L} \left(\frac{a_0}{2} + \sum_{k=1}^\infty d_k \cos(\phi_k)\right),
\end{align}
as all the harmonics are present at the output of the RC filter. In fact, when $\tau \to 0$, we have $v_{\rm o}(t) = R_L v_{\rm out}(t)$ since this leads to $H(f) \to R_L$.

\subsection{Discussion on Multisine Signals}
Multisine signals have particular interest for WPT as they boost the RF harvesting due to their high PAPR \cite{AG}. Thus, we now discuss how the proposed analytical framework can be extended to this case. Specifically, the source transmits an unmodulated $N$-tone multisine signal with zero phase arrangement and intercarrier frequency spacing $\Df$ \cite{PAN, IK}. In this case, the transmitted signal can be written as
\begin{align}
x(t) = U(t)\cos(2\pi f_c t),
\end{align}
where $U(t) = A \sin(N\pi \Df t)/\sin(\pi \Df t)$ and $f_c\gg \Df$. The envelope $U(t) = U(t+T)$ is a periodic function with $T = 2k/\Df$ and fundamental frequency $f_0 = \Df/2$. Observe that for $\Df \to 0$, we have $U(t) \approx AN$, which follows from the small-angle approximation. In other words, the Fourier series representation as well as the DC voltage is equal to the single sinewave case but scaled linearly by the number of tones $N$.

The derivation of the Fourier coefficients for the multisine case is out of the scope of this letter. However, we provide $a_0$ for $N=2$, since it characterizes the DC voltage at the output of the filter. Specifically, for the full-wave rectifier we get
\begin{align}
a_0 = \frac{8 f_c \Big(2 f_c - \Df \sin \Big(\frac{\pi}{4} \frac{\Df}{f_c}\Big)\Big)}{\pi(4 f_c^2 - \Df^2)}\cos \left(\frac{\pi}{4}\frac{\Df}{f_c}\right),
\end{align}
whereas for the half-wave rectifier we have
\begin{align}
a_0 = \frac{8 f_c^2}{\pi(4f_c^2-\Df^2)} \cos\bigg(\frac{\pi}{4}\frac{\Df}{f_c}\bigg).
\end{align}
The above can be derived using the methodology in the appendix. Observe that, in this case, the full-wave coefficient is greater than the half-wave one by a factor $2 - \Df \sin \Big(\frac{\pi}{4} \frac{\Df}{f_c}\Big)/f_c$. Moreover, for $f_c \to \infty$, we end up with the coefficients of Proposition \ref{prop1}.

\begin{figure}\centering
  \includegraphics[width=0.9\linewidth]{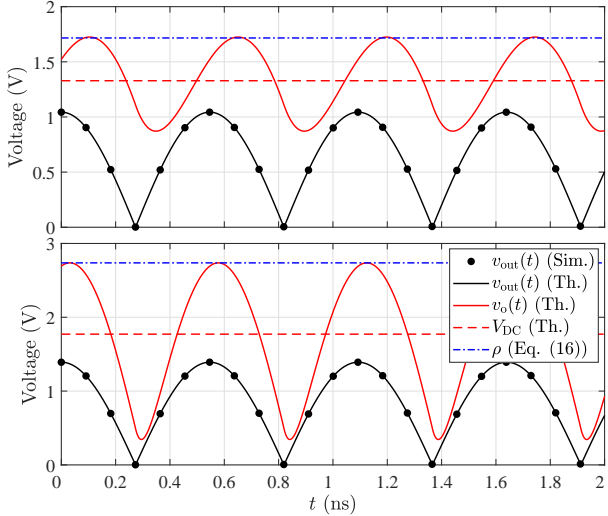}
  \caption{Output voltage versus $t$; top: $\fcut = 1$ GHz, bottom: $\fcut = 5$ GHz.}\label{fig3}
\end{figure}

\section{Numerical Results \& Conclusions}
We now validate our analytical approach for a single sinewave signal with computer simulations. The considered parameters are $A = 1$ V, $R_L = 2 ~\Omega$ and $f_c = 915$ MHz.

Fig. \ref{fig3} illustrates the output voltage in terms of the time instant $t$ for the full-wave rectifier with $\fcut = 1$ GHz (top sub-figure) and $\fcut = 5$ GHz (bottom sub-figure). Evidently, the proposed analytical framework captures the expected behavior of the rectifying circuit. Specifically, with a higher $\fcut$ (i.e., a lower $\tau$), the system achieves a higher DC voltage but there is also a larger ripple. It is also important to point out, that Fig. \ref{fig3} validates the analytical expression for the maximum ripple value. This is a critical aspect, as the maximum tolerance for the ripple is application-specific and is defined based on the desired objectives.

Fig. \ref{fig4} depicts the DC voltage with respect to $\fcut$ for the considered rectifier models. It is clear that for both rectifiers, $\Vdc$ converges to its maximum value as $\fcut$ increases. A faster convergence is attained for smaller values of the carrier frequency $f_c$. As expected, the full-wave rectifier achieves twice the DC values of the half-wave rectifier.

Future extensions of this work include the consideration of circuit imperfections (e.g. impedance mismatch), more realistic diode models as well as an investigation on the impact of fading and the phase shift on the energy harvesting.

\begin{figure}\centering
  \includegraphics[width=0.9\linewidth]{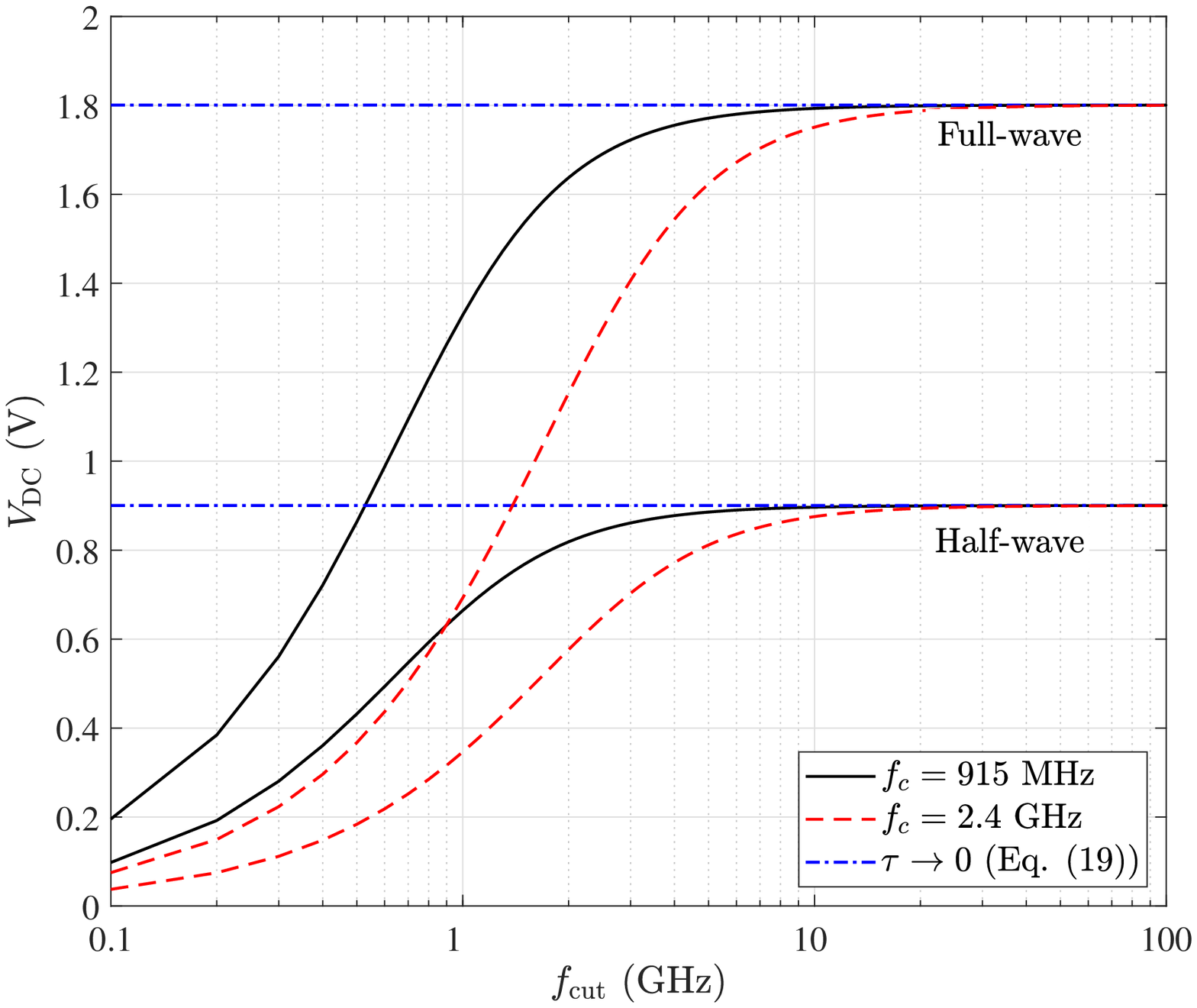}
  \caption{$\Vdc$ in terms of the cut-off frequency $f_{\rm cut}$.}\label{fig4}
\end{figure}

\section{Acknowledgments}
The authors would like to thank Kun Chen-Hu for providing the simulation results of Section \ref{circ_sim}.

\appendix
For the real signal $y(t) = \cos(2\pi f_c t)$, the Fourier coefficients (trigonometric Fourier series) are given by \cite{ISG}
\begin{align}
a_k = 2f_c \int_{-\frac{1}{2f_c}}^{\frac{1}{2f_c}} g(y(t))\cos(2\pi k f_c t) dt,
\end{align}
and
\begin{align}
b_k=2f_c \int_{-\frac{1}{2f_c}}^{\frac{1}{2f_c}} g(y(t))\sin(2\pi k f_c t) dt,
\end{align}
where recall that $g(x)=|x|$ (full-wave) or $g(x)=\max(0,x)$ (half-wave). We first evaluate the full-wave coefficients. In this case, we have
\begin{align}
a_k &= \frac{1}{\pi} \Bigg(2\int_0^\frac{\pi}{2} \cos (z) \cos(kz) dz - 2\int_\frac{\pi}{2}^\pi \cos (z) \cos(kz) dz\Bigg)\label{ak1}\\[-2mm]
&=\frac{2}{\pi} \int_{-\frac{\pi}{2}}^{\frac{\pi}{2}} \cos (z) \cos(kz) dz,
\end{align}
which follows from the transformation $z \to 2\pi f_c t$ and the fact that $\cos(z) > 0$ for $z \in (-\frac{\pi}{2},\frac{\pi}{2})$ and $\cos(z) < 0$, otherwise. The final expression of $a_k$, $k\neq 1$, can be deduced through the use of well-known trigonometric identities \cite{ISG}, whereas for $k=1$, \eqref{ak1} is equal to zero. In a similar way, we can evaluate the coefficient $b_k$ by
\begin{align}
b_k = \frac{2}{\pi} \int_{-\frac{\pi}{2}}^{\frac{\pi}{2}} \cos (z) \sin(kz) dz,
\end{align}
which it is easy to see that it reduces to $b_k = 0$ for all $k$.

The half-wave coefficients can be derived by using the methodology above. The main difference is the fact that, in this case, we have
\begin{align}
a_k &= \frac{1}{\pi}\int_{-\pi}^\pi \max(0,\cos(z)) \cos(kz)dz\nonumber\\
&= \frac{1}{\pi} \int_{-\frac{\pi}{2}}^{\frac{\pi}{2}} \cos(z) \cos(kz)dz,
\end{align}
which gives \eqref{coeff_half} and completes the proof.

\end{document}